\documentclass[a4paper,UKenglish]{lipics-v2018}

\usepackage{microtype}
\usepackage{xfrac}
\usepackage{enumitem}

\hideLIPIcs

\author{Ivor van der Hoog}{Department of Information and Computing Sciences, Utrecht University, the Netherlands}{i.d.vanderhoog@uu.nl}{}{Supported by the Netherlands Organisation for Scientific Research (NWO); 614.001.504.}
\author{Irina Kostitsyna}{Department of Mathematics and Computer Science, TU Eindhoven, the Netherlands}{i.kostitsyna@tue.nl}{}{}
\author{Maarten L\"{o}ffler}{Department of Information and Computing Sciences, Utrecht University, the Netherlands}{m.loffler@uu.nl}{}{Partially supported by the Netherlands Organisation for Scientific Research (NWO); 614.001.504.}
\author{Bettina Speckmann}{Department of Mathematics and Computer Science, TU Eindhoven, the Netherlands}{b.speckmann@tue.nl}{}{Partially supported by the Netherlands Organisation for Scientific Research (NWO); 639.023.208.}

\authorrunning{I. van der Hoog, I. Kostitsyna, M.  L\"{o}ffler, B. Speckmann} 

\Copyright{Ivor van der Hoog, Irina Kostitsyna, Maarten  L\"{o}ffler, Bettina Speckmann}



\ArticleNo{XXX}
\nolinenumbers
\category{}

\supplement{}

\acknowledgements{The authors would like to thank Jean Cardinal for insightful discussions.}

\EventEditors{}
\EventNoEds{2}
\EventLongTitle{35th International Symposium on Computational Geometry (SoCG 2019)}
\EventShortTitle{SoCG 2019}
\EventAcronym{SoCG}
\EventYear{2019}
\EventDate{}
\EventLocation{}
\EventLogo{socg-logo} 
\SeriesVolume{99}

\usepackage{lineno}
\usepackage[utf8]{inputenc}
\usepackage{amsmath}
\usepackage{amsfonts}
\usepackage{amsthm}
\usepackage{graphicx}
\usepackage{wrapfig}
\usepackage{algorithm}
\usepackage[noend]{algpseudocode}
\usepackage[colorinlistoftodos]{todonotes}
\usepackage{makecell}
\usepackage{xspace}
\usepackage{caption}
\usepackage{imakeidx}

\newcommand{\etal}{\textit{et al.}\xspace}
\modulolinenumbers[5]

\usepackage[utf8]{inputenc}
\usepackage{amsmath,amsfonts,amssymb,amsthm,MnSymbol}
\usepackage{graphicx}
\graphicspath{{./figures/}}

\let\emptyset\varnothing

\newcommand{\RR}{\ensuremath{\mathcal R}\xspace}

\newcommand{\Oh}{\ensuremath{\mathcal O}\xspace}

\newcommand{\T}{\ensuremath{e}\xspace}

\newcommand{\SO}{\ensuremath{A}\xspace}
\newcommand{\CS}{\ensuremath{\Gamma^\pi}\xspace}
\newcommand{\RA}{\ensuremath{\bot^{\pi}}\xspace}

\title{Preprocessing Ambiguous Imprecise Points}

\makeindex
\begin{document}

\maketitle

\begin{abstract}
Let $\RR = \{R_1, R_2, \ldots, R_n\}$ be a set of regions and let $
X = \{x_1, x_2, \ldots, x_n\}$ be an (unknown) point set with $x_i \in R_i$. Region $R_i$ represents the uncertainty region of $x_i$. We consider the following question: how fast can we establish order if we are allowed to preprocess the regions in \RR?
The \emph {preprocessing model} of uncertainty uses two consecutive phases: a preprocessing phase which has access only to $\RR$ followed by a reconstruction phase during which a desired structure on $X$ is computed. Recent results in this model parametrize the reconstruction time
by the \emph{ply} of $\RR$, which is the maximum overlap between the regions in $\RR$. We introduce the \emph{ambiguity} $\SO(\RR)$ as a more fine-grained measure of the degree of overlap in $\RR$.
We show how to preprocess a set of $d$-dimensional disks in $\Oh(n \log n)$ time such that we can sort $X$ (if $d=1$) and reconstruct a quadtree on $X$ (if $d\geq 1$ but constant) in $\Oh(\SO(\RR))$ time. If $\SO(\RR)$ is sub-linear, then reporting the result dominates the running time of the reconstruction phase. However, we can still return a suitable data structure representing the result in $\Oh(\SO(\RR))$ time.  

In one dimension, $\RR$ is a set of intervals and the ambiguity is linked to interval entropy, which in turn relates to the well-studied problem of sorting under partial information. The number of comparisons necessary to find the linear order underlying a poset $P$ is lower-bounded by the graph entropy of $P$.
We show that if $P$ is an interval order, then the ambiguity provides a constant-factor approximation of the graph entropy. This gives a lower bound of $\Omega(\SO(\RR))$ in all dimensions for the reconstruction phase (sorting or any proximity structure), independent of any preprocessing; hence   
 our result is tight. Finally, our results imply that one can approximate the entropy of interval graphs in $\Oh(n \log n)$ time, improving the $\Oh(n^{2.5})$ bound by Cardinal~\etal
 \end{abstract}
 

 \section{Introduction}

A fundamental assumption in classic algorithms research is that the input data given to an algorithm is exact. Clearly this assumption is generally not justified in practice: real-world data tends to have (measurement or labeling) errors, heterogeneous data sources introduce yet other type of errors, and ``big data'' is compounding the effects. To increase the relevance of algorithmic techniques for practical applications, various paradigms for dealing with uncertain data have been introduced over the past decades. Many of these approaches have in common that they represent the uncertainty, imprecision, or error of a data point as a \emph{disk} in a suitable distance metric which we call an uncertainty region. We focus on a fundamental problem from the realm of computation with uncertainties and errors: given a set of imprecise points represented by uncertainty regions, how much proximity information do the regions contain about the imprecise points?

\subparagraph*{Preprocessing model.} We study this problem within the preprocessing framework initially proposed by Held and Mitchell \cite{held2008triangulating}. In this framework we have a set $\RR = \{R_1, R_2, \ldots, R_n\}$ of regions and an point set $X = \{x_1, x_2, \ldots, x_n\}$ with $x_i \in R_i$
This model has 2 consecutive phases: a preprocessing phase followed by a reconstruction phase. In the preprocessing phase we have access only to $\RR$ and we typically want to preprocess $\RR$ in $\Oh(n \log n)$ time to create some linear-size auxiliary data structure which we will denote by $\Xi$. In the reconstruction phase, we have access to $X$ and we want to construct a desired output on $X$ using $\Xi$ faster than would be possible otherwise. 
L{\"o}ffler and Snoeyink~\cite{loffler2010delaunay} were the first to use this model as a way to deal with data uncertainty: one may interpret the regions $\RR$ as {\em imprecise} points, and the points in $X$ as their true (initially unknown) locations. This interpretation of the preprocessing framework
has been successfully applied to various problems in computational geometry~\cite {buchin2009delaunay,buchin2011delaunay,devillers2011delaunay,ezra2013convex, loffler2013unions,van2010preprocessing}. Several results restrict $\RR$ to be a set of disjoint (unit) disks in the plane, while others consider partially overlapping disks. Traditionally, the \emph{ply} $\Delta(\RR)$ of $\RR$, which measures the maximal number of overlapping regions, has been used to measure the degree of overlap, leading, for example, to reconstruction times of $O(n \log \Delta(\RR))$.

\begin{figure}[b]
\centering
\includegraphics{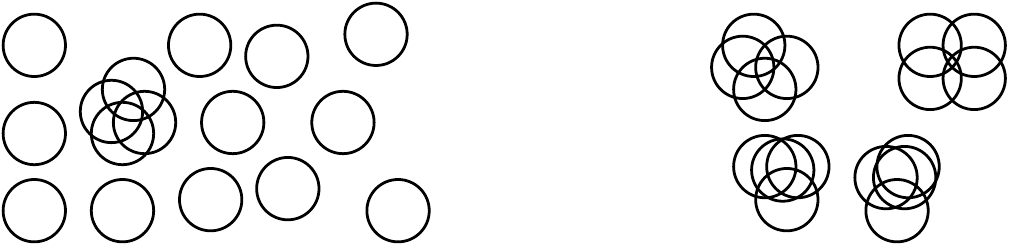}
\caption{Two sets of 16 disks each in the plane, both with a ply of 4. The ambiguity of the set on the right is four times as large as the ambiguity of the set on the left. }
\label{fig:lowply}
\end{figure}

The ply is arguably a somewhat coarse measure of the degree of overlap of the regions. Consider the following example: suppose that we have a collection of $\sqrt{n}$ disks in the plane that overlap in one point and that the remainder of $\RR$ is mutually disjoint (see Figure~\ref{fig:lowply} left). Then $\Delta(\RR) = \sqrt{n}$ and the resulting time complexity of the reconstruction phase is $\Oh(n \log n)$ even though it might be possible to achieve better bounds ($\RR$ is arguably not in a worst-case configuration for that given ply, see Figure 1 right).

\subparagraph*{Ambiguity.} We introduce the \emph{ambiguity} $\SO(\RR)$ as a more fine-grained measure of the degree of overlap in $\RR$. The ambiguity is based on the number of regions each individual region intersects (see Figure~\ref{fig:lowply}). We count this number with respect to particular permutations of the regions: for each region we  count only the overlap with regions that appear earlier in the permutation. A proper technical definition of ambiguity can be found in Section~\ref{sec:ambiguity}. 
We also show how to compute a $3$-approximation of the ambiguity in $\Oh(n \log n)$ time.  

\subparagraph*{Ambiguity and entropy.} In one dimension, $\RR$ is a set of intervals and the ambiguity is linked to interval (and graph) entropy (see Appendix~\ref{appx:entropy} for a definition), which in turn relates to the well-studied problem of sorting under partial information. Fredman~\cite{fredman1976good} shows that if the only information we are given about a set of values is a partial order $P$, and $e(P)$ is the number of {\em linear extensions} (total orders compatible with) of $P$, then we need at least $\Omega(\log \T(P))$ comparisons to sort the values. Brightwell and Winkler prove that computing the number of linear extensions $\T(P)$ is $\#P$-complete~\cite{brightwell1991counting}. Hence efforts have concentrated on computing approximations, most notably via the concept of graph entropy as introduced by K\"{o}rner~\cite{korner1973coding}. Specifically, Khan and Kim~\cite{kahn1995entropy} prove that $\log \T(P) = \Theta(n \cdot H(G))$ where $H(G)$ denotes the entropy of the \emph{incomparability graph} $G$ of the poset $P$. To the best of our knowledge there is currently no exact algorithm to compute $H(G)$. Cardinal~\etal~\cite{cardinal2013sorting} describe the fastest known algorithm to approximate $H(G)$, which runs in $\Oh(n^{2.5})$ time. See Appendix~\ref{appx:entropy} for a more in-depth discussion of sorting and its relation to graph entropy. 


We consider the special case where the partial order is induced by uncertainty intervals. We define the entropy $H(\RR)$ of a set of intervals as the entropy of their intersection graph (which is also an incomparability graph) using the definition of graph entropy given by K\"{o}rner. In this setting we prove that the ambiguity $\SO(\RR)$ provides a constant-factor approximation of the interval entropy (see Section~\ref{sec:ambiguity}). Since we can compute a constant-factor approximation of the ambiguity in $\Oh(n \log n)$ time, we can hence also compute a constant-factor approximation of the entropy of interval graphs in $\Oh(n \log n)$ time, thereby improving the result by Cardinal~\etal~\cite{cardinal2013sorting} for this special case.
 
\subparagraph*{Ambiguity and reconstruction.} Since $\Omega(\log \T(P))$ is a lower bound for the number of comparisons needed to complete $P$ into a total order, $\Omega(\SO(\RR))$ is a lower bound for the reconstruction phase in the preprocessing model when $\RR$ is a set of intervals and the goal is to sort the unknown points in $X$. This lower bound extends to higher dimensions and to proximity structures in general, independent of any preprocessing.

The ambiguity $\SO(\RR)$ ranges between $0$ and $\Theta(n \log n)$ for a set of $n$ regions $\RR$. If the value of $\SO(\RR)$ lies between $\Theta(n)$ and $\Theta(n \log n)$ then we can preprocess $\RR$ in $\Oh(n \log n)$ time and sort in $\Oh(\SO(\RR))$ time (in one dimension for arbitrary intervals) or build a quadtree in $\Oh(\SO(\RR))$ time (in all dimensions for unit disks).

If the ambiguity lies between $0$ and $\Theta(n)$, then reporting the results explicitly in $\Omega(n)$ time dominates the reconstruction time. But the ambiguity suggests that the information-theoretic amount of work necessary to compute the results should be lower than $\Theta(n)$. To capture this, we hence introduce a new variant of the preprocessing model, which allows us to return a pointer to an implicit representation of the results.

Specifically, in one dimension, $\RR$ is a set of intervals and we aim to return the sorted order of the unknown points in $X$. If, for example, all intervals are mutually disjoint, then $\SO(\RR) = \mathcal{O}(1)$ and we have essentially no time for the reconstruction phase. However, a binary search tree $T$ on $\RR$, which we can construct in $\Oh(n \log n)$ time in the preprocessing phase, actually captures all necessary information. In the reconstruction phase we can hence return a pointer to $T$ as an implicit representation of the sorted order. In Section~\ref{sec:sorting} we show how to handle arbitrary sets of intervals in a similar manner. That is, we describe how to construct in $\Oh(n \log n)$ time an auxiliary data structure $\Xi$ on \RR in the preprocessing phase (without access to $X$), such that, in the reconstruction phase (using $X$), we can construct a linear-size AVL-tree $T$ on $X$ in $\Oh(\SO(\RR))=\Oh(\log e(\RR) )$ time, which is tight. 

In all dimensions, we consider $\RR$ to be a set of unit disks and our aim is to return a quadtree $T$ on the points in $X$ where each point in $X$ lies in a unique quadtree cell. 
Note that in 2 dimensions, $T$ also allows us to construct e.g. the Delaunay triangulation of $X$ in linear time~\cite{buchin2009delaunay}.
However, we show that constructing such a quadtree explicitly in $\Oh(\SO(\RR))$ time is not possible, and the work necessary to distinguish individual points could dominate the running time and overshadow the detail in the analysis brought by the ambiguity measure. 
We hence follow Buchin~\etal~\cite{buchin2009delaunay} and use so-called \emph{$\lambda$-deflated quadtrees} which contain up to a constant $\lambda$ points in each leaf. From $T$ one can construct a quadtree on $X$ where each point lies in a unique quadtree cell in linear time. In Section~\ref{sec:quadtrees} we describe how to reconstruct a linear-size $\lambda$-deflated quadtree $T$ (with a suitable constant $\lambda$) in $\Oh(\SO(\RR))=\Oh(\log e(\RR) )$ time, which is tight (in fact, in one dimension our result also extends to non-unit intervals). 


\section{Ambiguity}\label{sec:ambiguity}

We introduce a new measure on a set of regions \RR to reflect the degree of overlap, which we call the \emph{ambiguity}.
The sequence in which we process regions matters (refer to Section~\ref{sec:lowerbound}), thus we distinguish between the \emph{$\pi$-ambiguity} defined on a given permutation of the regions in $\RR$, and the minimum ambiguity defined over all possible permutations. 
We demonstrate several properties of the ambiguity, and discuss its relation to graph entropy when $\RR$ is a set of intervals in one dimension. 
%
\subparagraph*{Processing permutation.}
Let $\RR$ be a set of $n$ regions and let $\RR^\pi = \langle R_{1}, R_{2}, \ldots, R_{n}\rangle$ (note that for all $i$, the region $R_i$ could be any region depending on the permutation $\pi$) be the sequence of elements in $\RR$ according to a given permutation $\pi$. Then we say that $\pi$ is a \emph{processing permutation} of $\RR$. Furthermore, let $\RR^\pi_{\le i} := \{R_j \mid j \leq i \}$ be the prefix of $\RR^\pi$, that is, the first $i$ elements in the sequence $\RR^\pi$. A permutation $\pi$ is \emph{containment-compatible} if $R_i \subset R_j$ implies $i<j$ for all $i$ and $j$ \cite{ft98}. When $\pi$ is clear from context, we denote $\RR^\pi$ by $\RR$.
\subparagraph*{Contact set (for a permutation $\pi$).}
For a region $R_i\in\RR^\pi$ we define its \emph{contact set} $\CS_i$ to be the set of regions which precede or are equal to $R_i$ in the order $\pi$, and which intersect $R_i$:
$\CS_i := \{R_j\in\RR^\pi_{\le i} \mid R_j\cap R_i\not=\emptyset\}$.
Note that a region is always in its own contact set. A region $R_i$ whose contact set $\CS_i$ contains only $R_i$ itself is called a \emph{bottom region} (refer to Figure~\ref{fig:intersectiongraph}). 
%
\subparagraph*{Ambiguity.} For a set of regions \RR and a fixed permutation $\pi$ we define the $\pi$-\emph{ambiguity} $\SO^\pi(\RR) := \sum_i \log |\CS_i|$ (with the logarithm to the base 2). Observe that bottom regions do not contribute to the value of the $\pi$-ambiguity. The ambiguity of $\RR$ is now the minimal $\pi$-ambiguity over all permutations $\pi$, $\SO(\RR) := \min_{\pi \in \Pi} \SO^\pi(\RR)$.
\begin{figure}[t]
\centering
\includegraphics[page=3]{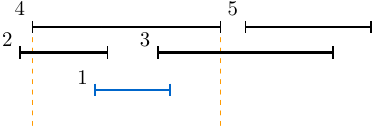}
\captionsetup{width=\linewidth}
\caption{A set of overlapping intervals with a permutation. In this figure $\CS_4 = \{R_1, R_2, R_3, R_4 \}$. In all figures, bottom intervals are indicated in blue (in this case this is only $R_1$).}
\label{fig:intersectiongraph}
\end{figure}
\subsection{Properties of ambiguity}
\label{sec:lowerbound}

We show the following properties of ambiguity: (1) the $\pi$-ambiguity may vary significantly with the choice of the processing permutation $\pi$, (2) in one dimension, the $\pi$-ambiguity for any containment-compatible permutation $\pi$ on a set of intervals \RR implies a 3-approximation on the entropy of the interval graph of \RR, and (3) the permutation that realizes the ambiguity is containment-compatible. Therefore in one dimension, the ambiguity of a set of intervals \RR implies a 3-approximation of the entropy of the interval graph of \RR.

We start with the first property: it is easy to see that the processing permutation $\pi$ has a significant influence on the value of the $\pi$-ambiguity
(refer to Figure~\ref{fig:ambiguity}). 
%
Even though $\pi$-ambiguity can vary considerably, we show that if we restrict the permutations to be containment-compatible, 
their $\pi$-ambiguities lie within a constant factor of the ambiguity. 

\begin{figure}[t]
\centering
\includegraphics[]{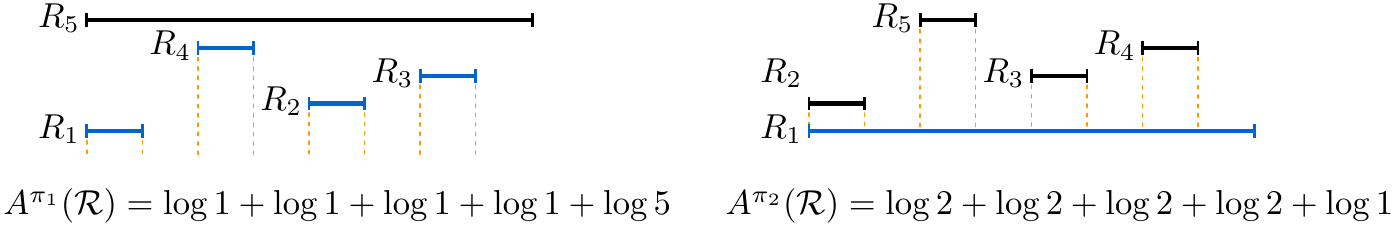}
\captionsetup{width=\linewidth}
\caption{An example of the $\pi$-ambiguity induced by two permutations $\pi_1$ (on the left) and $\pi_2$ (on the right) of the same five intervals. The $\pi_1$-ambiguity is $\log 5$ and the $\pi_2$-ambiguity is $4$.}
\label{fig:ambiguity}
\end{figure}

\subparagraph*{Interval entropy.}
The entropy of a graph $G$ was first introduced by K\"{o}rner \cite{korner1973coding}. Since then several equivalent definitions appeared \cite{simonyi1995graph}. We define the \emph{interval entropy} $H(\RR)$, for a set of intervals \RR, as the entropy of the intersection graph of $\RR$. While investigating the question of sorting an arbitrary poset, Cardinal \etal \cite{cardinal2013sorting} found an interesting geometrical interpretation of the poset entropy, which applies to our interval entropy: let a poset $P$ describe a set of (open) intervals $\RR$ combinatorially, that is, for each $R_i$ we know which intervals intersect $R_i$, are contained in $R_i$, contain $R_i$, and are disjoint from $R_i$. Denote by $E(\RR)$ the infinite set of sets of intervals on the domain $(0,1)$ (that is, each $\mathcal{I} \in E(\RR)$ is a set of intervals, where each interval $I_i \in \mathcal{I}$ has endpoints in $(0,1)$) which induce the same poset as $\RR$. Then Cardinal \etal prove the following lemma (see Figure~\ref{fig:entropy} for an illustration):
\begin{lemma}[\cite{cardinal2013sorting}, Lemma 3.2 paraphrased]
\label{lemma:cardinal}
\[
    H(\RR) =  \log n -  \min_{\mathcal{I} \in E(\RR)} \left\{ \frac{1}{n}  \sum_{I_i \in \mathcal{I}}  -\log |I_i|   \right \}\,.
\]
\end{lemma}
We show that the $\pi$-ambiguity for any containment-compatible $\pi$ is a 3-approximation of $n \cdot H(\RR)$. To achieve this we rewrite the lemma from Cardinal \etal in the following way,
\[
    H(\RR) =  \log n -  \min_{\mathcal{I} \in E(\RR)} \left\{ \frac{1}{n}  \sum_{I_i \in \mathcal{I}} \left(\log{n} -\log(n|I_i|) \right) \right\}= \max_{\mathcal{I} \in E(\RR)} \left\{ \frac{1}{n}  \sum_{I_i \in \mathcal{I}} \log(n|I_i|) \right \}\,.
\]
%
%
%
An embedding $\mathcal{I}$ gives each interval $I_i$ a size between $0$ and $1$. To simplify the algebra later, we re-interpret this size as the \emph{fraction} (weight) of the domain $(0,1)$ that $I_i$ occupies. We associate with each $\mathcal{I} \in E(\RR)$ a set of weights $W$ such that for all $i$, $w_i = |I_i|$; we write $W \sim E(\RR)$. 
From now on we consider embeddings on the domain $(0,n)$: an interval then has a size $n|I_i| = n w_i$.
The formula for the entropy becomes:
\begin{equation}
\label{eq:weights}
       H(\RR) = \frac{1}{n} \max_{W \sim E(\RR)} \left\{    \log \left( \prod_{w_i \in W}n\,w_i \right) \right \}\,. 
\end{equation}

\begin{figure}[t]
\centering
\includegraphics{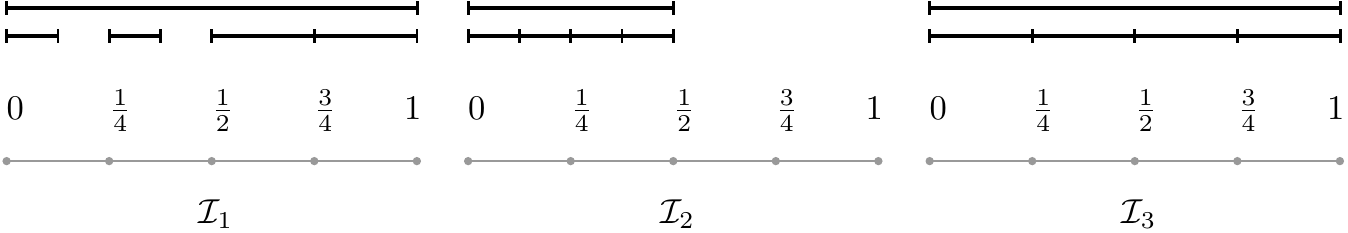}
\captionsetup{width=\linewidth}
\caption{Let $\RR$ be a set of five intervals, where four intervals are mutually disjoint and contained in one larger interval. We show three embeddings $\mathcal{I}_1, \mathcal{I}_2, \mathcal{I}_3 \in E(\RR)$ of these intervals on the domain $(0, 1)$ with the same combinatorial properties. Embedding $\mathcal{I}_1$ shows that $H(\RR) \ge \log 5 - \frac{1}{5} \log(1 \cdot  \frac{1}{8}\cdot \frac{1}{8}\cdot \frac{1}{4}\cdot \frac{1}{4} )$. $\mathcal{I}_2$ shows that $H(\RR) \ge \log 5 - \frac{1}{5} \log( \frac{1}{2}\cdot \frac{1}{8}\cdot \frac{1}{8}\cdot \frac{1}{8}\cdot \frac{1}{8})$ and $\mathcal{I}_3$ is the optimal embedding which shows that $H(\RR) = \log 5 -\frac{1}{5} \log( 1\cdot  \frac{1}{4}\cdot \frac{1}{4}\cdot \frac{1}{4}\cdot \frac{1}{4}) $.}
\label{fig:entropy}
\end{figure}

\subparagraph*{Ambiguity and entropy.} 
Next, we show that the interval entropy gives an upper bound on the ambiguity. The entropy of $\RR$ is the maximum over all embeddings on $(0,n)$, so any embedding of \RR on the domain $(0,n)$ gives a lower bound on $H(\RR)$. We will create an embedding with a corresponding weight assignment $W$ such that:
\begin{equation}
\label{eq:entropy}
    \SO^\pi(\RR) = \log \left( \prod_{R_i \in \RR} |\CS_i| \right) \le \log \left( \prod_{w_i \in W} (n\,w_i)^2 \right) \le 2n H(\RR)\,.
\end{equation}
We start with the original input embedding of $\RR$ and we sort the coordinates of all the endpoints (both left- and right-). 
To each endpoint $p$ we assign a new coordinate $\frac{k}{2}$ if $p$ is the $k$th endpoint in the sorted order (indexing from 0).
Thus, we obtain an embedding of \RR on $(0,n-\frac{1}{2})$. 
For any containment-compatible permutation $\pi$, the length of each interval $R_i$ in this embedding is at least $\frac{1}{2}|\CS_i|$, as each interval $R_i$ contains at least $|\CS_i| - 1$ endpoints of the intervals from its contact set in its interior. 
Also note that the distance between every right endpoint and the consecutive endpoint to the right is $\frac{1}{2}$. 
Thus, we can increase the coordinate of every right endpoint by $\frac{1}{2}$ and obtain an embedding of $\RR$ on $(0,n)$ with a corresponding weight assignment $W$, such that the length of each interval $R_i$ is at least $\frac{1}{2}(|\CS_i|+1)$.
This allows us to prove the following lemma:

\begin{lemma}
\label{lemma:lowerbound}
For any containment-compatible permutation $\pi$ of a set of intervals \RR, 

$\SO^\pi(\RR) \le  2n H(\RR)$.
\end{lemma}
\begin{proof}
Consider the embedding and corresponding weight assignment $W$ constructed above. Consider any containment-compatible permutation $\pi$. We split the intervals of \RR into four sets depending on the size of their contact set: let $A := \{ R_i \mid |\CS_i| = 1\}$, $B := \{ R_i \mid |\CS_i| = 2\}$, $C := \{ R_i \mid |\CS_i| = 3\}$ and $D := \RR \backslash \{A,B,C\}$. Let these sets contain $a, b, c$ and $d$ intervals respectively. 
Then, using Equation~(\ref{eq:weights}) for the entropy,
\begin{equation}
\label{eq:base}
2^{nH(\RR)} \ge 
\prod_{R_i\in A}\frac{|\CS_i|+1}{2} \prod_{R_i\in B}\frac{|\CS_i|+1}{2} \prod_{R_i\in C}\frac{|\CS_i|+1}{2} \prod_{R_i\in D}\frac{|\CS_i|+1}{2}
\ge 
\left(\frac{2}{2}\right) ^a \left(\frac{3}{2}\right)^b   \left(\frac{4}{2}\right)^c \left(\frac{4}{2}\right)^d.
\end{equation}
On the other hand,
\[
2^{nH(\RR)} \ge
\prod_{R_i\in\RR}\frac{|\CS_i|+1}{2} \ge \prod_{R_i\in A}|\CS_i| \prod_{R_i\in B}\frac{3}{4}|\CS_i| \prod_{R_i\in C}\frac{2}{3}|\CS_i| \prod_{R_i\in D}\frac{1}{2}|\CS_i|
= \left(\frac{3}{4}\right)^b   \left(\frac{2}{3}\right)^c \left(\frac{1}{2}\right)^d 2^{\SO^\pi(\RR)},
\]
as
\begin{table}
	$\frac{|\CS_i|+1}{2} = \begin{cases}
	1=|\CS_i|\,, & \text{ if } R_i\in A\,,\\
	\frac{3}{2}=\frac{3}{4}|\CS_i|\,, & \text{ if } R_i\in B\,,\\
	2=\frac{2}{3}|\CS_i|\,, & \text{ if } R_i\in C\,,
	\end{cases}\\
	\frac{|\CS_i|+1}{2} \ge  \frac{1}{2}|\CS_i|\,, \quad\text{ if } R_i\in D\,.$
\end{table}
Then, using Equation~(\ref{eq:base}) we get 
\[
2^{nH(\RR)}\cdot 2^{nH(\RR)} \ge \left(\frac{3}{2}\right)^b   \left(\frac{4}{2}\right)^c \left(\frac{4}{2}\right)^d \cdot
\left( \frac{3}{4} \right)^b \left( \frac{2}{3}\right)^c \left( \frac{1}{2}\right)^d 2^{\SO^\pi(\RR)} \ge 2^{\SO^\pi(\RR)}\,,
\]
and therefore
\[
2nH(\RR) \ge \SO^\pi(\RR)\,. \qedhere
\]
\end{proof}
We continue by showing that the ambiguity also gives an upper-bound for the interval entropy. Starting with a helper lemma:

\begin{lemma}
\label{lemma:helper}
Suppose $\RR$ is partitioned into two sets $X$ and $Y$ such that for each $R \in X, R' \in Y$, $R$ and $R'$ are disjoint. In any weight assignment $W$ that realizes $H(\RR)$, the intervals in $X$ together have length $|X|$ and the intervals in $Y$ together have length $|Y|$ on the domain $(0, n)$.
\end{lemma}

\begin{proof}
In Equation~(\ref{eq:weights}) we rewrote the formula for entropy in terms of weights: for any weight assignment $W \sim E(\RR)$, $w_i$ is the proportion that $R_i$ occupies on the domain, and we embedded $\RR$ on the domain $(0,n)$. 
We can similarly embed $\RR$ on the domain $(0, \lambda)$ for an arbitrary scalar $\lambda$. 
We define the \emph{relative entropy} of $\RR$ (refer to Figure~\ref{fig:relativeentropy} (top)) as:
\[
    H(\RR, \lambda) := \frac{1}{n} \max_{W \sim E(\RR)} \left \{  \log \left(  \prod_{w_i \in W}  \lambda\, w_i \right)\right \} \,.
\]
Observe that $H(\RR, n) = H(\RR)$ and that:
\begin{equation}
\label{eq:rewrite}
    \forall \lambda, \mu, \quad \mu\, w_i = \left( \frac{\mu}{\lambda}  \right) \lambda\, w_i \Rightarrow 
    2^{nH(X, \mu)} =   \left( \frac{\mu}{\lambda} \right)^{|X|} 2^{nH(X, \lambda)}\,.
\end{equation}
If the intervals in $X$ can occupy a width of at most $\lambda$, then it is always optimal to give the intervals in $Y$ a total width of $n-\lambda$ (since the entropy maximizes the product of the lengths of intervals in $X$ and $Y$). This implies:
\[
    2^{nH(X \cup Y)} = \max_{\lambda \in [0,n]} \{ 2^{nH(X, \lambda)} \cdot  2^{nH(Y, n - \lambda)} \} \,.
\]
See Figure~\ref{fig:relativeentropy} (bottom) for an illustration of the argument. If we now substitute Equation~(\ref{eq:rewrite}) into this equation we get that the maximum is realized if $\lambda = |X|$ which proves the lemma.
\end{proof}
\begin{figure}[t]
\centering
\includegraphics{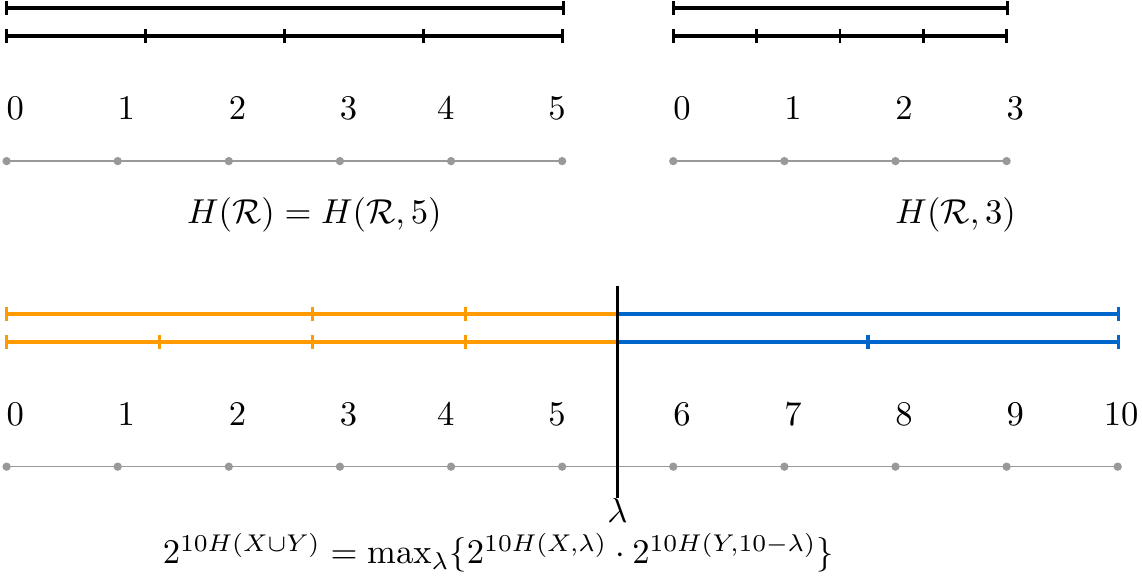}
\captionsetup{width=\linewidth}
\caption{(top left) A set $\RR$ of five intervals and their optimal embedding for the entropy relative to $\lambda = 5$. (top right) The optimal embedding of $\RR$ for the entropy relative to $\lambda = 3$. Observe that the proportion that each interval obtains of the domain is the same in both embeddings. (bottom) An illustration of the argument for Lemma~\ref{lemma:helper}: we see a set $X$ of 7 intervals and a set $Y$ of 3 intervals with the intervals in $X$ disjoint from the intervals in $Y$. If we vary $\lambda$, we vary the total width on which $X$ and $Y$ are embedded. The entropy is given by the maximal embedding and therefore found by optimizing $\lambda$.}
\label{fig:relativeentropy}
\end{figure}
\begin{lemma}
\label{lemma:upperbound}
Let $\pi$ be any containment-compatible permutation, then $nH(\RR) \le 3 \SO^{\pi}(\RR)$.
\end{lemma}
\begin{proof}
We defined $\RR_{\le i}$ as the prefix of $\RR$. We prove the lemma with induction on $i$. 
%
\[
\text{Induction Hypothesis: }  
\forall{j \le i} j H(\RR_{\le j}, j)\le 3 \SO^\pi(\RR_{\le j})\,.
\]
For $i=1$ both the lefthand and the righthand side are $0$. So we assume that the lemma holds for all $j \le i$ and we prove it for $j = i+1$. $H(\RR_{\le i+1}, i+1)$ is the relative entropy of $\RR_{\le(i+1)}$ on the domain $(0, i+1)$. We know that $3\SO^\pi(\RR_{\le (i+1}) = 3\SO^\pi(\RR_{\le i}) + 3 \log |\Gamma_{i+1}|$. We make a distinction between two cases: $|\Gamma_{i+1}| = 1$ or otherwise. 
If $|\Gamma_{i+1}| = 1$ then $R_{i+1}$ is disjoint from $R_{\le i}$. Lemma~\ref{lemma:helper} guarantees, that if we want to embed $R_{\le i} \cup \{R_{i+1} \}$ on $(0, i+1)$ that $R_{i+1}$ gets a size of $1$. The remaining intervals get embedded with a total width of $i$ which they already had in the previous iteration. So:
\[(i+1)H(\RR_{\le(i+1)}, i+1 ) = i H(\RR_{i}, i) + \log 1 \le 3\SO^\pi(\RR_{\le i}) + 3 \log |\Gamma_{i+1}| = 3 \SO^\pi(\RR_{\le(i+1)})\,.
\]
In the second case $|\Gamma_{i+1}|$ is at least $2$. The other intervals used to be optimally embedded on $(0, i)$ and are now embedded on $(0,i+1)$. So each of them expands with at most a factor $\frac{i+1}{i}$ or algebraically:
\[
    (i+1)H(\RR_{\le(i+1)}, i+1 ) \le  
    iH(\RR_{\le i}, i)
    + \log \left( \left( \frac{i+1}{i} \right)^i \right)
    + \log((i+1) w_{i+1}) 
    \le  
    iH(\RR_{\le i}, i)
    + \log e
    + \log((i+1) w_{i+1}) 
    \,.
\]
There are $i - |\Gamma_{i+1}|$ intervals disjoint from $R_{i+1}$ so Lemma~\ref{lemma:helper} guarantees that $(i+1) w_{i+1} \ge |\Gamma_{i+1}| \ge 2$. It follows that:
\[
nH(\RR_{\le(i+1)}, i+1 ) \le nH(\RR_{\le i}, i) + 3 \log |\Gamma_{i+1}|
\]
which implies the Lemma.
\end{proof}
Lemmas~\ref{lemma:lowerbound} and~\ref{lemma:upperbound} imply the following theorem.
\begin{theorem}
For any set of intervals $\RR$ in one dimension, for any containment-compatible permutation $\pi$ on $\RR$, $\SO^\pi(\RR)$ is a $3$-approximation of $nH(\RR)$.
\end{theorem}

\begin{corollary}
For any set of intervals $\RR$ in one dimension, the ambiguity $\SO(\RR)$ is a $3$-approximation of $n H(\RR)$.
\end{corollary}

\begin{proof}
The permutation which realizes the ambiguity of $\RR$ must always be containment-compatible. This is because swapping a region $R$ with a region $R'$ that contains $R$ in the permutation $\pi$ always improves the $\pi$-ambiguity.
\end{proof}
Let $\T(\RR)$ be the number of linear extensions of the poset induced by $\RR$. In the proof of Lemma~3.2~\cite{cardinal2013sorting} Cardinal \etal show that $\log \T(\RR) \le nH(\RR) \le 2 \log \T(\RR)$. This implies that the interval graph entropy is a lower-bound for constructing any unique linear order underlying a poset. 
Proximity structures depend on sorting~\cite{de2008computational}. Thus, we conclude: 
\begin{theorem}
Reconstructing a 
proximity structure on $\RR$ is lower-bounded by $\Omega(\SO(\RR))$.
\end{theorem}

\section{Sorting}
\label{sec:sorting}

Let $\RR = \{R_1, R_2, \ldots, R_n\}$ be a set of intervals and let $X = \{x_1, x_2, \ldots, x_n\}$ be a set of points (values) with $x_i \in R_i$. We show how to construct an auxiliary structure $\Xi$ on \RR in the preprocessing phase without using $X$, such that, in the reconstruction phase, we can construct a linear-size binary search tree $T$ on $X$ in $\Theta(\SO(\RR))$ time. To achieve this, we first construct a specific containment-compatible permutation $\pi$ of $\RR$, and then show how to maintain $\Xi$ when we process the intervals in this order.

\subsection{Level permutation}
\label{sec:levelpermutation}

\begin{figure}[t]
\centering
\includegraphics[]{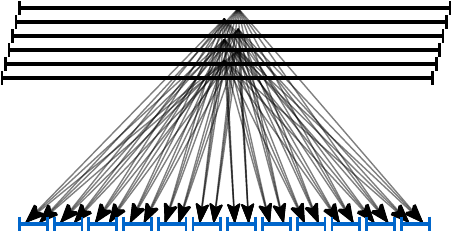}
\captionsetup{width=\linewidth}
\caption{A set of intervals with a containment graph with quadratic complexity.}
\label{fig:containmentgraph}
\end{figure}

We need a processing permutation $\pi$ of $\RR$ with the following conditions:
\begin{enumerate}[label=(\roman*)]
  \item $\pi$ is containment-compatible, 
  \item intervals containing no interval of $\RR$ come first and are ordered from right to left and
  \item we can construct $\pi$ in $\Oh(n \log n)$ time.
\end {enumerate}
In Section~\ref{sec:lowerbound} we showed that if condition (i) holds, the $\pi$-ambiguity is a lower-bound for sorting $X$. In Section~\ref{sec:algorithm} we show that condition (ii) is useful to reconstruct an AVL-tree on $X$ in $\Oh( \SO^{\pi}(\RR))$ time. Condition (iii) bounds the time used in the preprocessing phase.

Below, we define two natural partitions of $\RR$ based on the containment graph of $\RR$: the \emph {height partition} and the \emph {depth partition}. However, a permutation compatible with the height partition satisfies conditions (i) and (ii) but not (iii), and a permutation compatible with the depth partition satisfies conditions (i) and (iii) but not (ii).
Therefore, we define a hybrid partition, which we call the \emph{level partition}, which implies a permutation which does satisfy all three conditions, below.

\subparagraph*{Containment graph.}
For a set of intervals $\RR$, its containment graph $G(R)$ represents the containment relations on $\RR$. $G(R)$ is a directed acyclic graph where $R_i$ contains $R_j$ if and only if there is a directed path from $R_i$ to $R_j$ and all intervals $R \in \RR$ that are contained in no other interval of $\RR$ share a common root. The bottom intervals are a subset of the leaves of this graph.
Note that $G(\RR)$ can have quadratic complexity (Figure~\ref{fig:containmentgraph}).

\subparagraph*{Height and Depth partition.}
We define the \emph{height partition} as the partition of $\RR$ into $m$ levels $\mathcal{H} = H_1 \ldots H_m, H_i \subseteq \RR$ where all $R \in H_j$ have \emph{height} (minimal distance from $R$ to a leaf) $j+1$ in $G(\RR)$ or equivalently: the intervals in $H_{j+1}$ contain no intervals in $\RR \backslash H_{\le j}$ (Figure~\ref{fig:levels}). 
We analogously define the \emph{depth partition} as the partition of $\RR$ into $m$ levels $\mathcal{D} = D_1 \ldots D_m, D_i \subseteq \RR$ where all $R \in D_j$ have \emph{depth} (maximal distance from the root to $R$) $(m-j)$ in $G(\RR)$. 
Clearly any permutation compatible with $\mathcal{H}$ or $\mathcal{D}$ satisfies condition (i). 
All leaves of $G(\RR)$ have height $1$ so per definition are all in $H_1$ and thus any permutation compatible with $\mathcal{H}$ that sorts $H_1$ satisfies condition (ii). 
Clearly the same is not true for $\mathcal{D}$. On the other hand, in Lemma \ref{lemma:depth} we show how to construct $\mathcal{D}$ in $\Oh(n \log n)$ time. 
It is unknown whether the height partition can be created in $\Oh(n \log n)$ time (see Appendix~\ref{appx:dominance}). 
\begin{figure}[t]
\centering
\includegraphics[width=\linewidth]{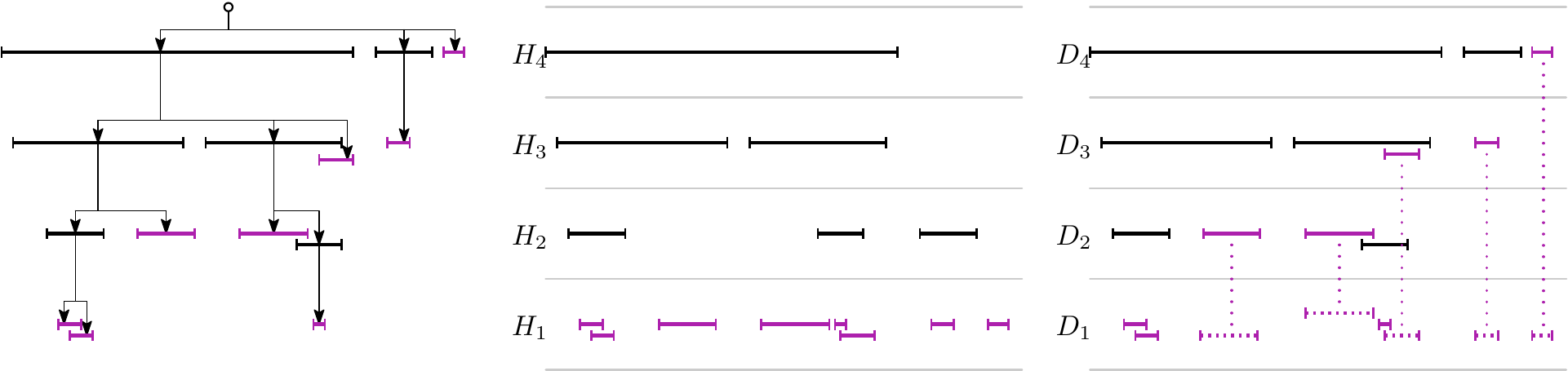}
\caption{(left) A set of intervals $\RR$ and the corresponding containment graph $G(\RR)$, leaves of $G(\RR)$ are purple. (middle) The height partition. (right) The depth partition.}
\label{fig:levels}
\end{figure}
\begin{lemma}
\label{lemma:depth}
For any set of intervals $\RR$ we can construct $\mathcal{D}$ in $\Oh(n \log n)$ time.
\end{lemma}

\begin{proof}
We iteratively insert intervals from left to right; refer to Appendix \ref{appx:depth}.
\end{proof}

\subparagraph*{Level partition.}

We now define the \emph{level partition}: a hybrid between $\mathcal{H}$ and $\mathcal{D}$: $\mathcal{L} =  L_1 \ldots L_m$, where all $R \in L_j$ have depth  $(m-j)$ in $G(\RR)$ except for the leaves of $G(\RR)$, which are in $L_1$ regardless of their depth. We can compute the level partition from $\mathcal{D}$ in $\Oh(n \log n)$ time by identifying all leaves of $G(\RR)$ with a range query. The \emph{level permutation} is the permutation where intervals in $L_i$ precede intervals in $L_j$ and where within each level the intervals are ordered from right to left. It can be constructed from $\mathcal{L}$ in $\Oh(n \log n)$ time by sorting.

\smallskip\noindent
Theorem~\ref{thm:permutation} follows directly from the preceding discussion.

\begin{theorem}\label{thm:permutation}
The level permutation satisfies conditions (i), (ii) and (iii).
\end{theorem}

\subsection{Algorithm}
\label{sec:algorithm}

We continue to describe a preprocessing and reconstruction algorithm to preprocess a set of intervals $\RR$ in $\Oh(n \log n)$ time such that we can sort $X$ in $\Theta(\SO(\RR))$ time.

\subparagraph*{Anchors.}

Let $\pi$ be the level permutation of \RR. In the preprocessing phase we build an AVL-tree $T$ on the bottom intervals. In the reconstruction phase, we insert each remaining $x_i \in X$ into $T$ in the order $\pi$ in $\Oh(\SO^\pi(\RR))$ time. This implies that for bottom intervals we are not allowed to spend even constant time and for each non-bottom interval $R_i$, we want to locate $x_i$ in $T$ in $\Oh(\log |\CS_i| )$ time. To achieve this, we supply every non-bottom interval $R$ with an \emph{anchor} denoted by $\RA(R_i)$. For a non-bottom interval $R \not \in L_1$, we define its anchor as an arbitrary interval contained in $R$. All intervals in $L_1$ are ordered from right to left, so for any non-bottom interval $R \in L_1$, its right endpoint is contained in the interval preceding it and we make this interval the anchor of $R$ (refer to Figure~\ref{fig:preprocessing}).

\subparagraph*{Preprocessing phase.}
The auxiliary structure $\Xi$ is an AVL-tree $T$ on the bottom intervals, augmented with a set of pointers leading from intervals to their anchors. We will implement $T$ as a leaf-based AVL-tree, i.e., where values are stored in the leaves, and inner nodes are decision nodes. Finally, we will use a doubly linked list to connect the leaves of the tree.

\begin{figure}[t]
\centering
\includegraphics[page=4, width =\linewidth]{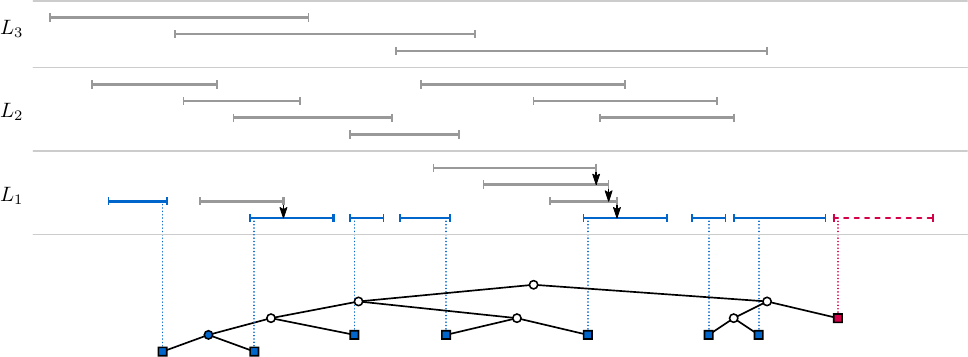}
\caption{The auxiliary structure $\Xi$. In the level $L_1$ all non-bottom intervals are shown their anchor. (top) A schematic representation of intervals in the level permutation $\pi$ (from bottom to top). (bottom) The Fibonacci tree $T$ containing the subset $X^b$ corresponding to the bottom intervals. Note that we added one dummy node in red. }
\label{fig:preprocessing}
\end{figure}

Let  $X^b \subset X$ be the points corresponding to bottom intervals. Bottom intervals are mutually disjoint and we can build an AVL-tree $T$ on $X^b$ without knowing their true values. Recall that a Fibonacci tree is a tree binary where for every inner node, its left subtree has a depth 1 greater than its right subtree. A Fibonacci tree is a valid AVL-tree and we construct the AVL-tree over $X^b$ as a Fibonacci tree where we add at most $|X^b|$ dummy leaves with value $\infty$ to ensure that the total number of nodes is a Fibonacci number. Refer to Figure~\ref{fig:preprocessing} for an example. We remove the bottom intervals from \RR and for each non-bottom interval $R$ we identify its anchor $\RA(R)$ and we supply $R$ with a pointer to $\RA(R)$. As the final step of the preprocessing phase we connect the leaves of $T$ in a doubly linked list.
To summarize: $\Xi$ consists of a graph of intervals connected by anchor pointers and an AVL-tree $T$. Each bottom interval is in $T$ and each non-bottom interval has a directed path to a node in $T$.

\begin{lemma}\label{lem:preprocessing}
We can construct the auxiliary structure $\Xi$ in $\Oh(n \log n)$ time.
\end{lemma}
\begin{proof}
The level partition and permutation can be constructed in $\Oh(n \log n)$ time and with it we get access to the intervals in $L_1$ sorted from right to left. We scan $L_1$ from right to left and for each interval $R \in L_1$ we either identify it as a bottom interval or to supply it with its anchor. We identify for each $R \not \in L_1$ its anchor in logarithmic time using a range query. We construct the Fibonacci tree on $X^b$ with leaf pointers in $\Oh(n \log n)$ time \cite{nievergelt1973binary}. 
\end{proof}

\subparagraph*{Reconstruction phase.}

During the reconstruction phase, we need to maintain the balance of $T$ when we insert new values. $T$ contains bottom intervals which we are not allowed to charge even constant time, so the classical amortized-constant analysis \cite{mehlhorn1986amortized} of AVL-trees does not immediately apply. Nonetheless we show in Appendix \ref{appx:avl}:
\begin{lemma}
\label{lemma:amortized}
Let $T$ be an AVL-tree where each inner node has two subtrees with a depth difference of 1. We can dynamically maintain the balance of $T$ in amortized $\Oh(1)$ time.
\end{lemma}

\begin{figure}[t]
\centering
\includegraphics[page=5, , width =\linewidth]{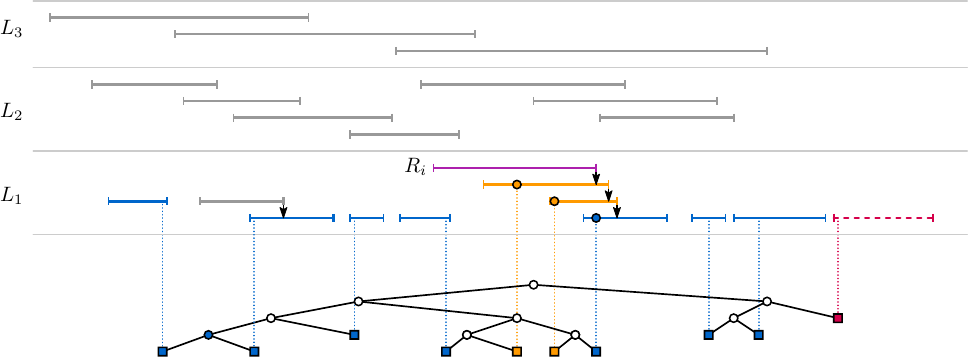}
\caption{ The tree $T$ from Figure~\ref{fig:preprocessing} after two iterations in the reconstruction phase. We inserted the true values of the two orange intervals. Note that an orange interval requested the true value of a bottom interval. At this iteration we want to insert the point $x_i$ of $R_i$ into $T$. $R_i$ is a non-bottom interval in $W_1$ so its anchor must be the interval preceding it. 
}
\label{fig:algo}
\end{figure}
\begin{theorem}\label{thm:algo}
Given $\Xi$, we can reconstruct an AVL-tree on $X$ in $\Theta(\SO^\pi(\RR))$ time.
\end{theorem}

\begin{proof}
Given $\Xi$ and the level permutation $\pi$ we want to sort the points in $X$ (insert them into $T$) in $\Oh(\SO^\pi(\RR))$ time. Because $T$ starts as a Fibonacci tree, Lemma~\ref{lemma:amortized} guarantees that we can dynamically maintain the balance of $T$ with at most $\Oh(\SO^\pi(\RR))$ operations.
The bottom intervals are already in $T$, thus we need to insert only the remaining $x_i \in X \backslash X^b$, in the order $\pi$, into $T$ in $\log |\CS_i|$ time plus some additional time which we charge to the anchor (each anchor will only get charged once).

Whenever we process a non-bottom interval $R_i$ we know that its anchor is already inserted in $T$. By construction, there are at most $\mathcal{O}(|\CS_i|)$ leaves in $T$ which have coordinates on the domain of $R_i$ (because these values can come only from intervals in the contact set of $R_i$). 
We know that we must insert $x_i$ next to one of these $\mathcal{O}(|\CS_i|)$ leaves in $T$. 
This means that if we have a pointer to any leaf on the domain of $R_i$, then we locate $x_i$ in $T$ with at most $O(\log |\CS_i|)$ edge traversals. During these traversals, we collapse each interval we encounter to a point.
We obtain such a pointer from $\RA(R_i)$. Assume $\RA(R_i) \subset R_i$. Then the leaf corresponding to $\RA(R_i)$ must lie on the domain of $R_i$. 
Otherwise, $R_i$ and $\RA(R_i)$ are both in the level $L_1$ (illustrated in Figure~\ref{fig:algo}) and $\RA(R_i) = R_{i-1}$ and must contain the right endpoint of $R_i$. With a similar analysis, $R_{i-1}$ can locate the right endpoint of $R_i$ in $T$ in $\Oh(\log |\CS_{i-1}|$ time. In both cases we found a leaf of $T$ in $R_i$ and locate $x_i$ in $T$ in $O(\log |\CS_i|)$ time. Each interval in $L_1$ has a unique anchor, so each anchor in $L_1$ is charged this extra work once.
\end{proof}
\section{Quadtrees}\label{sec:quadtrees}

Let $\RR = \{R_1, R_2, \ldots, R_n\}$ be a set of unit intervals in a bounding box (interval) $\mathcal{B}$ (we discuss how to extend the approach later) and let $X = \{x_1, x_2, \ldots, x_n\}$ be a set of points (values) with $x_i \in R_i$. We show how to construct an auxiliary structure $\Xi$ on \RR in the preprocessing phase without using $X$, such that, in the reconstruction phase, we can construct a linear-size quadtree $T$ on $X$ in $\Theta(\SO(\RR))$ time.
We recall several standard definitions.

\subparagraph*{Point quadtrees.}
Suppose that we have a $d$-dimensional point set $X$ in a bounding hypercube $\mathcal{B}$. 
A quadtree on $(\mathcal{B}, X)$ is defined as follows:  \emph{split} operator is an operator that splits any $d$-dimensional hypercube into $2^d$ equal-sized hypercubes called \emph{cells}. We recursively split $\mathcal{B}$ until each point $p \in P$ lies within a unique cell \cite{samet1984quadtree}. A $\lambda$-deflated quadtree is a more relaxed quadtree where $\mathcal{B}$ is split until each leaf cell contains at most $\lambda$ points \cite{buchin2011delaunay}.

\begin{figure}[t]
\centering
\includegraphics {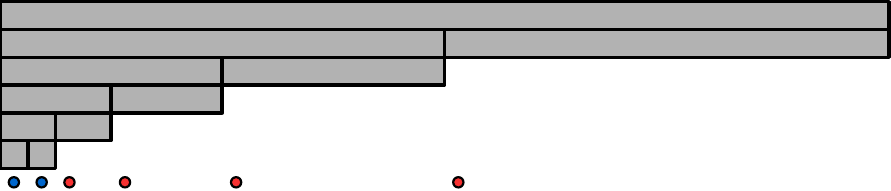}
\caption{A set of points $\RR$ where the quadtree on $\RR$ has linear depth. If the blue points lie very close, the quadtree on $\RR$ needs unbounded complexity. }
\label{fig:imbalance}
\end{figure}

\subparagraph*{Region quadtrees.}

Let $\RR$ be a set of $d$-dimensional disks in a bounding hypercube $\mathcal{B}$. Let $\mathcal{T}(\mathcal{B})$ be the infinite set of possible quadtree cells on $\mathcal{B}$. For each $R_i \in \RR$, we define its \emph{storing cell} denoted by $C_i$ as the largest cell in $\mathcal{T}(\mathcal{B})$ that is contained in $R_i$  and contains the center of $R_i$ \cite{loffler2013dynamic}. $T_i$ is the subtree induced by $C_i$. The \emph{neighborhood} of $R_i$ is the set of possible cells $C \in \mathcal{T}(\mathcal{B})$ with size $|C_i|$ that are intersected by $R_i$. We consider the quadtree $T$ on $\RR$ to be the unique compressed quadtree where for each $R_i \in \RR$, its neighborhood is in $T$.

\begin{figure}[b]
\centering
\includegraphics{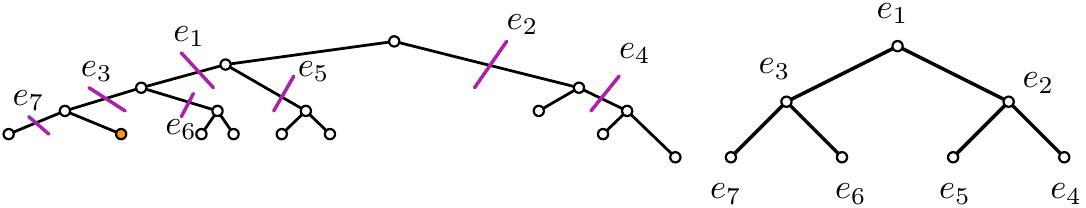}
\caption{(left) A tree $T$ with recursive centroid edges. (right) The corresponding edge-oracle tree $E$. The orange leaf is a subtree of $T$ and its corresponding node in $E$ is $e_3$. }
\label{fig:edgeoracle}
\end{figure}

\subparagraph*{Edge oracle tree.}
Depending on $\mathcal{B}$ and $X$, the quadtree on $(\mathcal{B}, X)$ does not necessarily have logarithmic depth (Figure~\ref{fig:imbalance}) thus, point location in $T$ is non-trivial. Har-Peled \cite{har2011geometric} introduced a fast point-location structure (later dubbed \emph{edge-oracle tree} \cite{loffler2013dynamic}) for any quadtree $T$. The edge-oracle tree $E$ is created through \emph{centroid decomposition}. Any tree with bounded degree $\delta$ has at least one \emph{centroid edge} which separates a tree of $n$ nodes into two trees with at least $\frac{n}{\delta}$ and at most $n -\frac{n}{\delta}$ nodes each. Moreover, one of these 2 trees is a \emph{subtree} of $T$ (a tree induced by a node as a root). 
For any subtree $T'$ of $T$, we define its \emph{corresponding node} in $E$ (edge in $T$) as the lowest node in $E$ which splits $T$ into two parts, one of which contains $T'$ and the other contains the root of $T$. This node must exist, is unique and the subtree containing $T'$ has $\Oh(|T'|)$ nodes (refer to Figure~\ref{fig:edgeoracle}).

Given a query point $q$, we can find the leaf cell $C_q$ that contains $q$ in the following way: each decision node $v$ of $E$ has 2 children where 1 child node $w$ corresponds to a subtree $T_w$ of $T$. We test whether $q$ is contained in $w$ in $\mathcal{O}(1)$ time by checking the bounding box of $T_w$.

\vspace{1em}\noindent
We wish to preprocess $\RR$ such that we can reconstruct a linear-size $\lambda$-deflated quadtree $T$ for $X$ with pointers between leaves. However, $T$ does not necessarily have linear size and dynamically maintaining pointers between leaves is non-trivial.
To achieve this, one needs to maintain a \emph{compressed} and \emph{smooth} quadtree $T$ (refer to Appendix \ref{appx:basic} for details) and Hoog \etal \cite{hoog2018dynamic} show how to dynammically maintain a smooth compressed quadtree with constant update time. We will build such a quadtree augmented with an edge-oracle tree initialized as a Fibonacci tree. We proceed analogously to the approach in Section~\ref{sec:sorting}.
 
\subsection{1-dimensional quadtrees on unit-size intervals}

We show how to construct an auxiliary structure $\Xi$ on \RR without using $X$, such that we can construct a $2$-deflated quadtree $T$ on $(\mathcal{B},X)$  in $\Theta(\SO(\RR))$ time.

\subparagraph*{Preprocessing phase.}

The auxiliary structure $\Xi$ will be a smooth compressed quadtree $T$ on the intervals $\RR$ augmented with an edge-oracle tree $E$ on $T$, anchor pointers, and a containment-compatible processing permutation $\pi$ of $\RR$.  Given $T$, we initialize $E$ as a Fibonacci tree, possibly adding dummy leaves\footnote{We may need to allow parents of leaves of $T$ to have a single dummy leaf.}. We supply each $R_i$ with a pointer to the node in $E$ corresponding to $T_i$ and we call this its \emph{anchor} $\RA(R_i)$.

\begin{lemma}
The auxiliary structure $\Xi$ can be constructed in $\Oh(n \log n)$ time.
\end{lemma}

\begin{proof}
Hoog \etal~\cite{hoog2018dynamic} show that for any set of $d$-dimensional disks $\RR$, its smooth compressed quadtree $T$ on $\RR$ with corresponding edge-oracle tree $E$ can be constructed in $\Oh(n \log n)$ time and that this tree has a worst-case constant update time. We turn $E$ into a Fibonacci tree by inserting at most $\Oh(n)$ dummy leaves in $\Oh(n \log n)$ time in total.
\end{proof}

\subparagraph*{Reconstruction phase.}

By construction, each leaf in $T$ intersects at most 2 bottom intervals of $\RR$ (since these are mutually disjoint). Therefore, we can construct a $2$-deflated quadtree on $X$ by inserting each $x_i \in X \backslash X^b$ in the order $\pi$ into $T$. We observe the following:

\begin{lemma}
\label{lemma:size}
When we process an interval $R_i \in \RR$, $R_i$ intersects $\Oh(|\CS_i|)$ leaf cells of $T$.
\end{lemma}

\begin{proof}
There can be at most 2 bottom intervals (left and right) of $R_i$ whose neighborhood intersects $R_i$. All the other leaves on the domain of $R_i$ are caused by either already processed points on the domain of $R_i$ or are dummy nodes. For each dummy node there is a corresponding non-dummy node also on the domain of $R_i$.
\end{proof}

\begin{lemma}
\label{lemma:quadtreetime}
When we process an interval $R_i$, we can locate, for any point $q \in R_i$, the leaf $C_q \in T$ which contains $q$ in $\Oh( \log |\CS_i|)$ time.
\end{lemma}

\begin{proof}
 If $C_q \in T_i$ then $R_i$ has an anchor to $T_i$ and from this anchor we locate $C_q$ in $\Oh(\log |\CS_i|)$ time. Suppose $C_q$ is to the left of $T_i$.  We locate the left-most leaf of $T_i$ in  $\Oh(\log |\CS_i|)$ time and traverse its neighbor pointer. The neighboring cell must lie in a subtree $T_q$ neighboring $T_i$ with $\Oh(|\CS_i|)$ nodes and this tree must contain $C_q$ (Lemma~\ref{lemma:size}). We now have a pointer to a node in $T_q$ and from this node we locate $C_q$ in $\Oh(\log |\CS_i|)$ time.
\end{proof}

\begin{theorem}
\label{thm:quadtrees}
Given $\Xi$, we can construct a $2$-deflated quadtree on $X$ in $\Theta(\SO^\pi(\RR))$ time.
\end{theorem}

\begin{proof}
Given $\Xi$ and any containment-compatible permutation $\pi$, we want to insert $X$ into $T$ in $O(\SO^\pi(\RR))$ time. An insertion in $T$ creates $2$ additional leaves in $T$ (and therefore also in $E$) and Lemma~\ref{lemma:amortized} guarantees that we can dynamically maintain the balance of $E$ with at most $O(\SO^\pi(\RR))$ operations. If we only consider the point set $X^b \subset X$ corresponding to the bottom intervals then $T$ is already a $2$-deflated quadtree on $X^b$ independent of where the points of $X^b$ lie in their uncertainty intervals. Therefore, we only need to insert the remaining $x_i \in X \backslash X^b$, in the order $\pi$, into $T$ in $\log |\CS_i|$ time (potentially collapsing some of the bottom intervals when necessary). Using Lemma~\ref{lemma:quadtreetime} we can locate the quadtree leaf $C_{x_i}$ that contains $x_i$ in $\Oh( \log |\CS_i|)$ time. This leaf is intersected by at most 2 bottom intervals, which we collapse into points whose location we locate in constant time using the leaf pointers. Thus each non-bottom interval inserts at most 3 points into $T$ in $\Oh( \log |\CS_i|)$ time.
\end{proof}

\subsection{Generalization}

If we stay in one dimension, then the result of Theorem \ref{thm:quadtrees} in fact generalizes to the case where $\RR$ is a set of arbitrary intervals since Lemma \ref{lemma:size} and \ref{lemma:quadtreetime} do not depend on the intervals being unit size. However, the result also generalizes to the case where $\RR$ is a set of unit-size disks in $d$ (constant) dimensions: first of all, any permutation of $\RR$ is containment-compatible. If the disks are unit size then each disk  intersects at most $K_d$ bottom disks where $K_d$ is the kissing number so Lemma \ref{lemma:size} generalizes. For any disk $R_i \in \RR$, recall that $T_i$ was the subtree of the storing cell of $R_i$. Any point $q \in R_i$ must lie in the \emph{perimeter} of $T_i$ which consists of at most $\Oh(5^d)$ subtrees of size $\Oh(|\CS_i|)$ therefore, Lemma \ref{lemma:quadtreetime} also generalizes.
The result is even more general: this approach works for any collection $\RR$ of unit-size \emph{fat} convex regions similar to, e.g. \cite{buchin2009delaunay}.
Interestingly, generalizing the result of Theorem~\ref{thm:quadtrees} both to higher dimensions and to non-unit regions at the same time is not possible: in Appendix~\ref {appx:disks} we show that, independent of preprocessing, reconstructing a $\lambda$-deflated quadtree has a lower bound of $\Omega (\log n)$, which could be more than $\SO(\RR)$.

\section{Conclusion}

We introduced the ambiguity $\SO(\RR)$ of a set of regions $\RR$ as a more fine-grained measure of the degree of their overlap. We applied this concept to uncertainty regions representing imprecise points. In the preprocessing model we show that the ambiguity is a natural lower bound for the time complexity of the reconstruction of any proximity structure. We achieved these results via a link to the entropy of partial orders which is of independent interest. If the regions are intervals in 1D we show how to sort in $\Theta(\SO(\RR))$ time, if the regions are unit balls in any dimension we show how to reconstruct quadtrees $\Theta(\SO(\RR))$ time. 


In the future we plan to investigate if our results can be generalized to other promixity structures such as Delaunay triangulations, minimum spanning trees, and convex hulls. In principle it is possible to convert quadtrees into all of these structures in linear time \cite{loffler2012triangulating}. However, it is not clear how to do so, when working with an implicit representation of the results in the case that $\SO(\RR)$ is sub-linear.

\bibliographystyle{abbrv}
\bibliography{hoog59}

\clearpage
\appendix

\section{Entropy of comparability and incomparability graphs} 
\label{appx:entropy}

K\"{o}rner \cite{korner1973coding} introduce the notion of the entropy of a graph. Let for any graph $G$, $\mathcal{A}_G$ be the space of independent sets of $G$. $\mathcal{A}_G$ is a convex subspace of $[0,1]^n$ where each integer-coordinate point in the space represents an independent subset of $G$. Let $\vec{x} = (x_1, x_2 \ldots x_n)$ be any (real-valued) point in $\mathcal{A}_G$. K\"{o}rner defines the graph entropy of $G$ as: $H(G) := \frac{1}{n} \max_{\vec{x} \in \mathcal{A}_G} \sum_i - \log x_i$ and this function is inspired by Shannon entropy.

Let $P$ be an arbitrary poset. The comparability graph $G_P$ of $P$ is the graph where there is an edge between $a,b \in P$ if $a$ and $b$ are comparable. The incomparability graph of $P$ is the graph where there is an edge $a,b \in P$ if $a$ and $b$ are incomparable and it is denoted by $\bar{G}_P$ since this is the complement of $G_P$.
Khan and Kim \cite{kahn1995entropy} define the entropy $H(P)$ of a poset $P$ as the entropy of $G_P$.
The more natural quantity to consider, however, is the 
entropy of the {\em incomparability} graph of $P$, which Khan and Kim denote by $H(\bar{P})$ (note that $H(P)+H(\bar{P})=\log n$). They continue to show that the time it takes to sort a poset $P$ is lower-bounded by $\Omega(nH(\bar{P}))$.

Cardinal \etal \cite{cardinal2013sorting} further investigate how to sort posets using this notion of entropy. They note that certain posets $P$ are induced by a set of intervals $\RR$; they call these \emph{interval orders}. Moreover, they show for every poset $P$, there exists an interval order $P'$ with $H(P) = H(P')$ (and hence also $H(\bar{P}) = H(\bar{P'})$). This allows them to approximate $H(\bar{P})$ for any poset $P$, by searching for a corresponding $P'$. 


\section{Building the depth partition}
\label{appx:depth}

We present the proof of Lemma~\ref{lemma:depth} in Section~\ref{sec:levelpermutation}, which states that for any set of intervals $\RR$ we can construct the {\em depth partition} $\mathcal{D}$ in $\Oh(n \log n)$ time.

\begin{proof}
To construct $D_1 \ldots D_m$ we process the intervals of \RR sorted by their \emph{left} endpoints from left to right. For each level $D_i$ we maintain the value $r_{i}$ as the maximum of the right endpoints of the intervals in $D_{i}$ and we maintain the invariant that $r_{i+1} \ge r_{i}$. Let $m$ be the (unknown) maximal level. Initially, we have $D_{m}$ as the empty set, no other sets and $r_{m} = - \infty$. We insert the first interval into $D_{m}$ and set  $r_{m}$ to be the right endpoint of the interval. 

\begin{figure}[b]
\centering
\includegraphics[page=1]{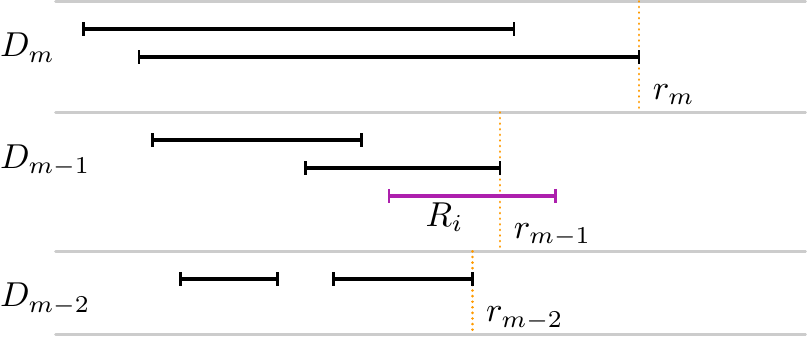}
\caption{An iteration of constructing the
depth partition. Intervals in black are already inserted. In this example, there are currently three levels, and a new interval $R$ is being inserted: $r_{m}$ is the only value greater than the right endpoint of $R$, thus $R$ is inserted into $D_{m-1}$.}
\label{fig:level-order}
\end{figure}

We then construct the remaining partition by iterating over the intervals in their sorted order. Consider the iteration where we are inserting an interval $R$ (refer to Figure~\ref{fig:level-order}). Let there be $\ell$ levels at this iteration $(D_m \ldots D_{m-\ell})$. We compare the right endpoint of $R$ denoted by $r$ with the values $r_m, r_{m-1}, \ldots, r_{m-\ell}$.
We find the minimal $j$ such that $r_j > r$ using binary search. All intervals in $D_j$ have a left endpoint left of $R$, so $R$ must be contained in an interval in $D_j$ and we therefore insert $R$ in the level $D_{j-1}$ and update $r_{j-1}$. This gives a partition where all intervals in a level $D_j$ have depth $m - j$ in the containment graph $G(\RR)$.
\end{proof}

\section{Building the height partition}
\label{appx:dominance}

In Section~\ref{sec:levelpermutation}, we introduced the {\em height partition} $\mathcal{H}$ as a natural partition of a set of intervals which would suit our needs, except for the fact that it is unclear how to compute it efficiently. We briefly expand on this here.

\begin{lemma}
\label{lemma:strongpartition}
For any set of intervals $\RR$, we can construct $\mathcal{H}$ in $\Oh(n \log^2 n)$ time.
\end{lemma}
\begin{proof}
Observe that an interval $(x_1,x_2)$ contains $(y_1,y_2)$ if and only if $x_1 \le y_1 \wedge y_2 \le x_2$. We use this information plus a 3-dimensional dynamic range tree \cite{de2008computational} to construct the height partition. We sort the intervals from narrow to wide and insert them into the correct level in this order. The least wide interval $(a,b)$ cannot contain an interval of $\RR$ so we store this interval in $S_1$ and we insert it in the dynamic  range tree as the 3-dimensional point $(a, b, 1)$.

Consider the iteration where we process an interval $(c, d)$. By this time we have already processed all intervals which could be contained in $(c,d)$.  We query the range tree with the following range: $(c, \infty) \times (-\infty, d) \times (-\infty, \infty)$ and we find the interval in this range with the maximal $z$-coordinate in $\Oh(\log^2 n)$ time. This gives us the interval $(e, f)$ which of all intervals contained in $(c,d)$, is stored in the highest level $S_j$. Thus, $(c,d)$ contains no intervals in $\RR \backslash S_{\le j}$ and must be stored in level $S_{j+1}$. Lastly we insert the point $(c,d, j+1)$ into the range tree in $\Oh(\log^2 n)$ time and we continue the iteration.
\end{proof}

Let for an interval $(c,d)$, $R_{\mid (c,d)}$ be the intervals in $\RR$ that are contained in $(c,d)$. During the construction of the height partition we want for $(c,d)$ to find the interval in $R_{\mid (c,d)}$ that is stored in the highest level. We project each interval $(a,b) \in R_{\mid (c,d)}$ to the point $(a,b,j)$ where $j$ is the level of $(a,b)$. We then perform a 3-dimensional range query on the range: $(c, \infty) \times (-\infty, d) \times (-\infty, \infty)$ to find the interval on this domain with the maximal $z$-coordinate. This leads to an interesting open problem which we will call \emph{dynamic $2.5$-queries}:

Let $\RR$ be a set of $n$ intervals where each interval $R_i$ has a weight $w_i$. Can we dynamically maintain a linear-size data structure on $\RR$, with $\Oh(\log n)$ update time that can answer the following query in logarithmic time: for an interval $(c,d) \in \RR$, what is the interval in $R_{\mid (c,d)}$ with the maximal weight?

Range queries with the range: $(q_1, \infty) \times (q_2, \infty) \ldots \times (q_k, \infty)$ are called $k$-dimensional \emph{dominance queries}. Range and dominance queries have 3 variants: reporting, counting and max. Here in the first 2 variants the goal is to report or count all the points within the range and where with the latter the goal is to return the maximal point within the range for some definition of maximal.

The query we posed above lies somewhere between a 2-dimensional and a 3-dimensional dominance max-query since the third dimension does not have a specified range. 2-dimensional dominance queries can be solved with a dynamic linear-size data structure with $\log n$ update and query time. Saxena \cite{saxena2009dominance} shows how static 3-dimensional dominance reporting can be solved using $\Oh(n \log n)$ space and construction time with $\Oh(\log n + S)$ query time where $S$ is the size of the output. This approach can be easily adapted to provide the element with the maximal $z$-coordinate instead but there is few hope for linear space and $\Oh(\log n)$ update time. At a similar time, Afshani \cite{afshani2008dominance} proposed a static data structure that answers 3-dimensional dominance reporting with $\Oh(n \log n)$ construction time, linear space and $\Oh(\log n + S)$ query time where $S$ is the size of the output. However, their approach cannot be adapted to dominance-max queries since it is based on an amortized-analysis using the size $S$. Our problem has more ``freedom'' than a full 3-dimensional range query.

\section{Basic properties of quadtrees}
\label{appx:basic}

Here we recall two basic properties of $d$-dimensional quadtrees, in particular compression and smoothness.
A quadtree $T$ on a real-valued point set $X$ is \emph{compressed} so that the quadtree has linear space. A smooth quadtree is a quadtree where each leaf is comparable in size to its adjacent leaves and smoothness is a prerequisite for having neighbor pointers between adjacent leaves.

\subparagraph*{Compression.}

Depending on $\mathcal{B}$ and the point set $X$, the quadtree on $(\mathcal{B}, X)$ does not always have linear size (Figure~\ref{fig:imbalance}). If two points of $X$ lie close, we potentially require an unbounded number of splits before they lie within a unique cell. The remedy for this is to use $\alpha$-\emph{compressed quadtrees}~\cite{har2011geometric}. For any quadtree $T$, if there exists a path of cells $v_1, v_2 \dots v_i$ with $i$ greater than a constant $\alpha$, such that the cells on this path contain only points which lie in $v_i$, then we do not explicitly construct this path. Instead, $v_1$ directly becomes the parent cell of the much smaller $v_i$. For any point set $P$, there exists a linear-size $\alpha$-compressed quadtree and this can be made dynamic with $\Oh(\alpha)$ update time~\cite{har2011geometric}.

\subparagraph*{Smoothness.}

A quadtree is {\em smooth} if each leaf is comparable in size to its adjacent leaves. It has been long recognized that smooth quadtrees are useful in many applications~\cite {BernEpGi94},
and smooth quadtrees can be computed in linear time (and have linear complexity) from their non-smooth counterparts~\cite[Theorem~14.4]{de2008computational}. 
In the previous sections, we used pointers between leaves of our AVL-tree to facilitate fast point location. Suppose you want to dynamically maintain a $d$-dimensional quadtree $T$, with $d > 1$, where each leaf in $T$ has pointers to adjacent leaves. Bennet and Yap \cite{BY17} show that this is possible with constant update time if and only we maintain a smooth quadtree. Hoog \etal \cite[Theorem 20]{hoog2018dynamic} show how to dynamically maintain a smooth compressed quadtree with constant update time.

\section{Amortized-constant balancing of an AVL tree}
\label{appx:avl}

In Section~\ref{sec:algorithm}, we need to maintain a dynamic balanced AVL tree in amortized constant time per update, even when the initial tree contains more elements than the number of updates we will perform.
Hence, a traditional amortization scheme which assumes the tree is build using update operations, does not immediately apply.
We show here that if we construct an initial AVL tree which is as {\em imbalanced} as possible (while remaining a valid AVL tree), we can in fact perform all future updates in contant time amortized only over the number of future updates.

We present the proof of Lemma~\ref{lemma:amortized}, which states that when $T$ is an AVL tree for which every internal node has two subtrees with a depth difference of 1, we can dynamically maintain the balance of $T$ in amortized $\Oh(1)$ time.

\begin{proof}[Proof of Lemma \ref{lemma:amortized}]

An AVL tree is a binary search tree where each inner node $v \in T$ has a balance constant $b(v) \in \{-1, 0, +1\}$. If $b(v) = -1$, the left subtree of $v$ has a depth 1 greater than the right subtree, if $b(v) = 0$ then both subtrees of $v$ have equal size and else $b(v) = +1$. 
Let $x$ be a value that we want to insert into $T$ and $v$ be the leaf that would become the parent of $x$. An insertion replaces $v$ with a decision node that is the parent of $x$ and $v$. Melhorn and Tsakalidis \cite{mehlhorn1986amortized} note that an insertion can recursively change balance constants all the way to the root. In particular, they observe that an insert will change a path of $0$'s into $+1$ or $-1$ until it reaches a node $\bar{v}$ where it will change $b(\bar{v})$ from $\{-1, +1\}$ to $0$ with a so called \emph{terminating action}: either the insert triggers a rotation around $\bar{v}$, or the insert guarantees that the left and right subtree of $\bar{v}$ are equal size. In both cases, the parent of $\bar{v}$ does not have to adjust its balance constant and the balancing process terminates. 

Our tree $T$ started out a tree where all the balance constants were $-1$ and $+1$. This means that whenever we insert a new value in $T$, it might change arbitrarily many $0$-nodes into $+1$ or $-1$, but for each of these nodes, there was a unique insert that created the $0$ value. This implies that performing all the balance changes takes amortized constant time.
\end{proof}

\section{$d$-dimensional quadtrees of arbitrary disks}
\label{appx:disks}

In Section \ref{sec:quadtrees} we show how to preprocess a set $\RR$ of uncertainty regions consisting of unit disks, such that in the reconstruction phase we can construct a $K_d$-deflated quadtree on the underlying point set $X$ in $\Theta(\SO(\RR))$ time. One key observation we make, is that if $\RR$ is a set of $d$-dimensional unit disks (for any metric), then any permutation $\pi$ on $\RR$ is a containment-compatible permutation. 
Suppose that $\RR$ is an arbitrary set of $d$-dimensional disks, then it is not clear if a containment-compatible permutation can be computed in $\Oh(n \log n)$ time (In Section \ref{sec:sorting} we used range queries to find this permutation but this is not fast enough in $d \ge 3$). Even if it would be possible to find a containment-compatible permutation $\pi$, then it still would not be possible to reconstruct a $\lambda$-deflated quadtree of $X$ in $\SO(\RR)$ time if the ambiguity of $\RR$ is below $\log n$:

\begin{theorem}
Let $\lambda$ be a constant. There exists a collection of discs in two dimensions $\RR$ with an $\SO(\RR) = \lambda$ with a corresponding point set $X$ such that reconstructing a $\lambda$-deflated quadtree on $X$ takes at least $\Omega(\log n)$ time.
\end{theorem}

\begin{proof}
We prove the theorem by constructing an example for arbitrary $\lambda$ illustrated by Figure~\ref{fig:2dlowerbound}.
\end{proof}

\begin{figure}[t]
\centering
\includegraphics{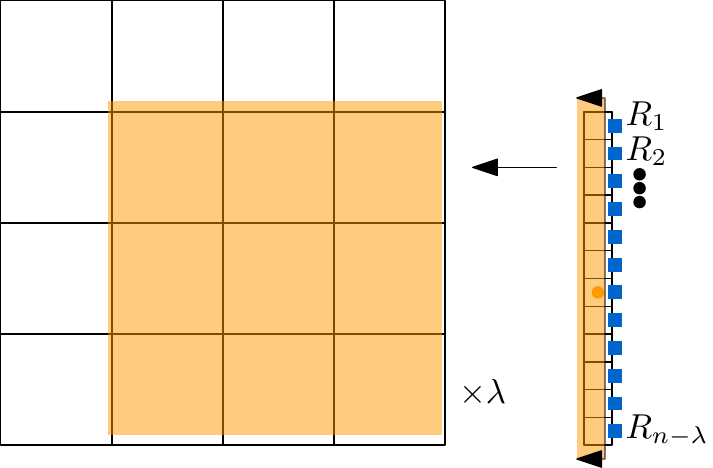}
\caption{This figure shows $\lambda$ identical disks $\Lambda \subset \RR$ in the $L_1$-metric and their neighborhood. The construction for the euclidean metric is the same but harder to illustrate. The arrow zooms in on the border of the $\lambda$ disks. There we see $n - \lambda$ mutually disjoint tiny disks, whose neighborhoods are all intersected by $\Lambda$. Suppose that all the $\lambda$ true values of $\Lambda$ lie within the same quadtree cell in the neighborhood of a blue disk, and that the true value of the blue disk also lies in that cell. Then there now exists a quadtree cell in $T$, where there are $\lambda + 1$ points of $X$ in a single leaf cell. If we want to construct a $\lambda$-deflated quadtree on $X$, then our only option is to locate and split this cell. But this location reduces to binary search on the $y$-coordinate of these true locations. Thus, constructing the $\lambda$-deflated quadtree on $X$ is lower-bounded by $\Omega(\log n)$ even though the ambiguity is constant. }
\label{fig:2dlowerbound}
\end{figure}

\end{document}